\documentclass[10pt, journal, letterpaper]{IEEEtran}

\usepackage{enumitem,calc}
\usepackage{amsmath,graphicx,cite}
\usepackage{cite,graphicx,epsfig,wrapfig,amsfonts,amssymb,threeparttable,color,url}
\usepackage{amsthm}
\usepackage{subfigure}
\usepackage{subfig}
\usepackage{epsf}
\usepackage{booktabs}
\usepackage{algorithm}
\usepackage{multirow}
\usepackage[noend]{algpseudocode}
\usepackage{float}
\usepackage{bm}
\usepackage{footnote}
\usepackage{fixmath}
\newtheorem{theorem}{Theorem}
\newtheorem{lemma}{Lemma}

\theoremstyle{definition}

\usepackage{amssymb}
\usepackage{pifont}

\newcommand{\MUS}{Multi-UAV Network }
\newcommand{\Board}{BoaRD }

\theoremstyle{remark}
\begin{document}
	
\title{Optimizing Number, Placement, and Backhaul Connectivity of Multi-UAV Networks}

\author{
Javad Sabzehali, Vijay K. Shah, Qiang Fan, Biplav Choudhury, Lingjia Liu, and Jeffrey H. Reed
\thanks{This work is partially supported by Office of Naval Research under Grant N00014-19-1-2621 and Commonwealth Cyber Initiative (CCI), an investment in the advancement of cyber R\&D, innovation, and workforce development. For more information about CCI, visit {www.cyberinitiative.org}

Javad Sabzehali, Biplav Choudhury, Lingjia Liu, and Jeffrey H. Reed are with Wireless@VT, The Bradley Department of ECE at Virginia Tech, Blacksburg, VA 24061 USA (e-mails: \{jsabzehali, biplavc, ljliu, reedjh\}@vt.edu).

Vijay K. Shah is with Department of Cybersecurity Engineering, George Mason University, Fairfax, VA 22030 USA (e-mail: vshah22@gmu.edu).

Qiang Fan is with Qualcomm, San Jose, CA 95110 USA (e-mail: qiangfan29@gmail.com). 
}
}

\maketitle

\begin{abstract} 
Multi-Unmanned Aerial Vehicle (UAV) Networks is a promising solution to providing wireless coverage to ground users in challenging rural areas (such as Internet of Things (IoT) devices in farmlands), where the traditional cellular networks are sparse or unavailable. A key challenge in such networks is the 3D placement of all UAV base stations such that the formed Multi-UAV Network (i) utilizes a minimum number of UAVs while ensuring -- (ii) backhaul connectivity directly (or via other UAVs) to the nearby terrestrial base station, and (iii) wireless coverage to all ground users in the area of operation. This joint \textit{Backhaul-and-coverage-aware Drone Deployment} (BoaRD) problem is largely unaddressed in the literature, and, thus, is the focus of the paper. We first formulate the BoaRD problem as Integer Linear Programming (ILP). However, the problem is NP-hard, and therefore, we propose a low complexity algorithm with a provable performance guarantee to solve the problem efficiently. Our simulation study shows that the Proposed algorithm performs very close to that of the Optimal algorithm (solved using ILP solver) for smaller scenarios, where the area size and the number of users are relatively small. For larger scenarios, where the area size and the number of users are relatively large, the proposed algorithm greatly outperforms the baseline approaches -- backhaul-aware greedy and random algorithm, respectively by up to $17\%$ and $95\%$ in utilizing fewer UAVs while ensuring $100\%$ ground user coverage and backhaul connectivity for all deployed UAVs across all considered simulation setting.
\end{abstract}


\maketitle
\vspace{-0.1in}
\section{Introduction}
Unmanned Aerial Vehicles (UAVs), or drones, have found numerous applications in recent years in the industry, government, and commercial fields, such as telecommunications, rescue operations, safety aerial sensing, safety surveillance, package delivery, and precision agriculture~\cite{mozaffari2019tutorial}. In particular, UAVs have attracted significant attention as the key enablers of end-to-end wireless communications, owing to their small size, positioning flexibility, and agile autonomy. One of the critical applications of UAVs is to provide wireless connectivity in remote/rural areas to wireless devices, such as Internet of Things (IoT) devices, where the conventional cellular wireless coverage is sparse or largely unavailable~\cite{s19092157}. Multiple UAVs can collaborate to form a \textit{Multi-UAV Network} and provide end-to-end wireless communication services to all ground users (e.g., IoT devices, etc.) in the area of operation. However, there are several challenges in doing so. 
 
Deploying a \MUS in order to provide wireless coverage to all ground users in a certain area of operation incurs enormous investment for network providers. Thus, \textit{it is critically important to utilize minimum number of UAVs} and facilitate the engineering application of Multi-UAV Networks. As the ground user distribution in a typical rural or remote area exhibits spatial and temporal dynamics, it is favorable to place UAVs close to more users and increase the number of users covered by deployed UAVs. Moreover, in order to ensure end-to-end wireless connectivity between the terrestrial base station (BS) and any UAV (via other UAVs as relay hops), it is important that all the deployed UAVs in the formed Multi-UAV Network maintain a backhaul connectivity among UAVs (in other words, there exists at least one path from a UAV to any other UAV in the formed Multi-UAV Network) \cite{9422729,8422376,9364991,9484638}.

Several works on UAV-assisted wireless communications have investigated optimal 3D placement of UAVs with a variety of objective functions such as ground user coverage maximization~\cite{al2014optimal}, maximizing average rate under worst-case bit error rate threshold~\cite{8642333}, energy-efficient communication maximization~\cite{alzenad20173}, etc. with applications in coverage and capacity enhancement. Please refer to Related Works for detailed discussion of related works. However, there is no work that studies the joint optimization of number, 3D placement and backhaul connectivity of Multi-UAV Networks. 

In this paper, we first formulate the above  \textit{Backhaul-and-coverage-aware Drone Deployment} (BoaRD) optimization problem as an integer linear programming (ILP) problem, considering communication channel constraints for all the access links and backhaul links, and show that it is NP-Hard. Following this, we propose a low computational complexity algorithm, using graph theoretic concepts, for solving the BoaRD problem with a provable performance guarantee. 

The key contributions of the paper are as follows.

\begin{itemize}[leftmargin=*]

\item We investigate the problem of UAV 3D placement such that the formed \MUS (i) employs a minimum number of UAVs, (ii) ensures wireless coverage to all ground users in the area of operation, and finally, (iii) guarantees that all UAVs in the \MUS are directly or indirectly (via one or more UAV relay hops) connected to the nearby terrestrial BS, i.e., the formed \MUS maintains an end-to-end backhaul connectivity among UAVs.

\item We formulate the joint optimization problem of number, placement, and backhaul connectivity of Multi-UAV Networks, termed, \textit{Backhaul and coverage-aware Drone Deployment} (BoaRD) problem as an Integer Linear Programming (ILP) optimization problem. We show that the Geometric Set Cover problem is reducible to the BoaRD problem, and prove that it is NP-Hard. To solve the problem efficiently, we use graph theoretic concepts and propose a low computational complexity algorithm with a provable performance guarantee. 

\item We solve the BoaRD problem using ILP solver (namely, Gurobi optimizer) and compare our proposed algorithm against it for smaller scenarios, where the area size of considered operation area is small (up to $10 \times 10$ sq. Km)  and number of users is few (e.g., up to $60$ users). The results shows that the Proposed algorithm performs quite well and utilizes close to optimal number of UAVs for smaller scenarios.

\item For larger scenarios (with larger operation area sizes (e.g., $50 \times 50$ sq Km.) and hundred's of ground users),  our large-scale simulations demonstrate that the Proposed algorithm significantly outperforms the two baselines -- Backhaul-aware Greedy (described in Section \ref{secResults}) and Random algorithms in utilizing fewer number of UAVs, by up to $17\%$ and $95\%$, respectively, across varying number of ground users, area sizes and SNR thresholds for backhaul links. Moreover, the Proposed algorithm outperforms the basic greedy solution (with no backhaul) by up to $13\%$ in high user density scenarios at low backhaul SNR thresholds (i.e., minimum SNR threshold for a successful connection between two UAVs), which further corroborates the efficacy of the Proposed algorithm.
\end{itemize}

The remainder of this paper is organized as follows. Section \ref{secRelated} reviews the related works. Section \ref{secSystem} presents the system model. In Section \ref{secProblem}, we formulate the BoaRD problem as an ILP optimization problem and prove its NP-hardness. Section \ref{secPreliminaries} discusses the key intuitions behind the Proposed solution, which are at the basis of the Proposed algorithm discussed in Section \ref{secSolution}. In Section \ref{secPerformanceProof}, we analyze the performance guarantee of the Proposed algorithm. Section \ref{secResults} discusses the experimental results. Finally, Section \ref{secConclusion} presents the concluding remarks.

\vspace{-0.1in}

\section{Related Works}
\label{secRelated}

Recent years have witnessed a surge of works such as in \cite{ matolak2012air,s19092157,8642333, al2014modeling, matolak2016air, yanmaz2013achieving, mozaffari2017mobile, bor2016efficient, al2014optimal, mozaffari2016efficient, alzenad20173,kalantari2016number, wang2020placement, you20193d, cui2020adaptive, xiong2019task, hu2019uav, huang2020trajectory,9467352,pawar2019performance, yi2019modeling, nguyen2020joint, B_Morteza_Stochastic_2020, 9419071, mayor_deploying_2019,kalantari2017backhaul,9521913,ansari2019soarnet, lyu2017placement, 9462529, 9314048, 9541037, 9497328, 9378782, 9364991, 8482444, 9484638} (as well as the comprehensive overview \cite{mozaffari2019tutorial}~\cite{fotouhi2019survey}) that attempts to address several research challenges in UAV communications and networking. Existing literature has focused on the following research directions --  (i) air-to-ground channel modeling~\cite{matolak2012air, al2014modeling, matolak2016air, yanmaz2013achieving}, (ii) UAV trajectory optimization~\cite{you20193d, cui2020adaptive, xiong2019task, hu2019uav, huang2020trajectory,9467352}, (iii) performance analysis of UAV-enabled wireless networks~\cite{pawar2019performance, yi2019modeling, nguyen2020joint, B_Morteza_Stochastic_2020}, and (iv) optimal placement of UAVs as flying base stations with a variety of objective functions such as energy-efficient communication maximization, sum-rate maximization, ground user coverage maximization, and maximize average rate under worst-case bit error rate threshold~\cite{s19092157, al2014optimal, 8642333, alzenad20173, mozaffari2017mobile, bor2016efficient, mozaffari2016efficient, kalantari2016number, wang2020placement, kalantari2017backhaul,9419071, mayor_deploying_2019,9521913,ansari2019soarnet, lyu2017placement, 9462529, 9314048, 9541037, 9497328, 9378782, 9364991, 8482444, 9484638}, with applications in coverage and capacity enhancement of 4G/5G cellular networks, flying ad-hoc networks (FANETS), and flying base stations for post-disaster situations, mmWave Communications, among others.

Though there exists a rich and growing body of literature on UAV communications, including optimal placement of UAVs, \textit{there is scant work on the design, analysis, and optimization of backhaul connectivity of Multi-UAV Networks while optimizing the placement of UAVs for ground user coverage.} Bor-Yaliniz \textit{et al.} \cite{bor2016efficient} proposed a 3-D placement algorithm to cover the maximum number of ground users given the fixed number of UAVs. Zhang \textit{et al.} \cite{9314048} optimized the 3D position of UAVs aiming to jointly minimize the number of UAVs and the coverage rate. Lin \textit{et al.} \cite{9541037} proposed an adaptive UAV deployment scheme to optimize UAV locations for more user coverage, and less communication energy consumption. Zhang \textit{et al.} \cite{9497328} studied the joint 3D deployment of the UAVs and power allocation problem to maximize the throughput of the UAV base station system. However, the mentioned works only focused on the access link in the drone-assisted network without considering the backhaul link. Ansari \textit{et al.} \cite{ansari2019soarnet}  proposed a drone communications framework, in which free space optical links are employed to serve as the backhaul link between UAV and ground base stations. They used the free space link to transfer data and energy to the UAV simultaneously and thus provision high-speed backhaul as well as prolong the UAV's flight. Lyu \textit{et al.} \cite{lyu2017placement} aimed to minimize the number of UAV-mounted mobile base stations (MBS) needed to provide wireless coverage for a group of distributed ground terminals (GTs), with the assumption that MBSs are backhaul-connected via satellite links. Neeto \textit{et al.} \cite{9530185} proposed an approach to partition the available resource between access links and a backhaul link by optimizing the placement of a single UAV in the area of operation. Hu \textit{et al.} \cite{9509753} maximized the system uplink throughput by jointly optimizing the UAV altitude, power control, and bandwidth allocation between the backhaul and access links. Pham \textit{et al.} \cite{9378782} aimed to maximize the sum rate achieved by ground users by jointly optimizing the UAV placement, spectrum allocation, and power control. Their scenario consists of a single UAV connected to a macro base station by a backhaul link and serves several ground users via access links. Iradukunda \textit{et al.} \cite{9364991} investigated the problem of maximizing the worst achievable rate for ground users by optimizing the UAV placement and power allocation, and bandwidth allocation. They considered that several users are connected to a single UAV via access links, which is connected to a macro base station via a backhaul link, by incorporating non-orthogonal multiple access (NOMA) scheme. Santos \textit{et al.} \cite{9448698} provided an approach to optimally place the UAVs as gateways in the area to address the dynamic traffic demand of access points, which itself is based on dynamic attributes of users. Nguyen \textit{et al.} \cite{8482444} aimed to maximize the user sum rate in a UAV-assisted cellular network by jointly optimizing the location of UAVs, the transmit bemformer at UAVs and a macro cell base station, and the decoding order of the NOMA-successive interference cancellation on wireless backhaul transmissions. They considered a number of UAVs directly backhauling to the macro cell base station and forming a two-hop backhaul connectivity. Dai \textit{et al.} \cite{9484638} aimed to maximize the end-to-end throughput in a UAV-assisted cellular network consisting of a single user, a macro cell base station, and a fixed number of UAVs. They derived the optimal position of UAVs considering power control, and orthogonal frequency schemes.

As evidenced, none of the existing work have investigated the joint optimization problem of number, placement, and backhaul connectivity of Multi-UAV Networks, which is the key focus of our work.

\section{Multi-UAV Network Model}
\label{secSystem}

\begin{figure}[h!]
\begin{center}
\vspace{-0.1in}
  \includegraphics[scale=0.31]{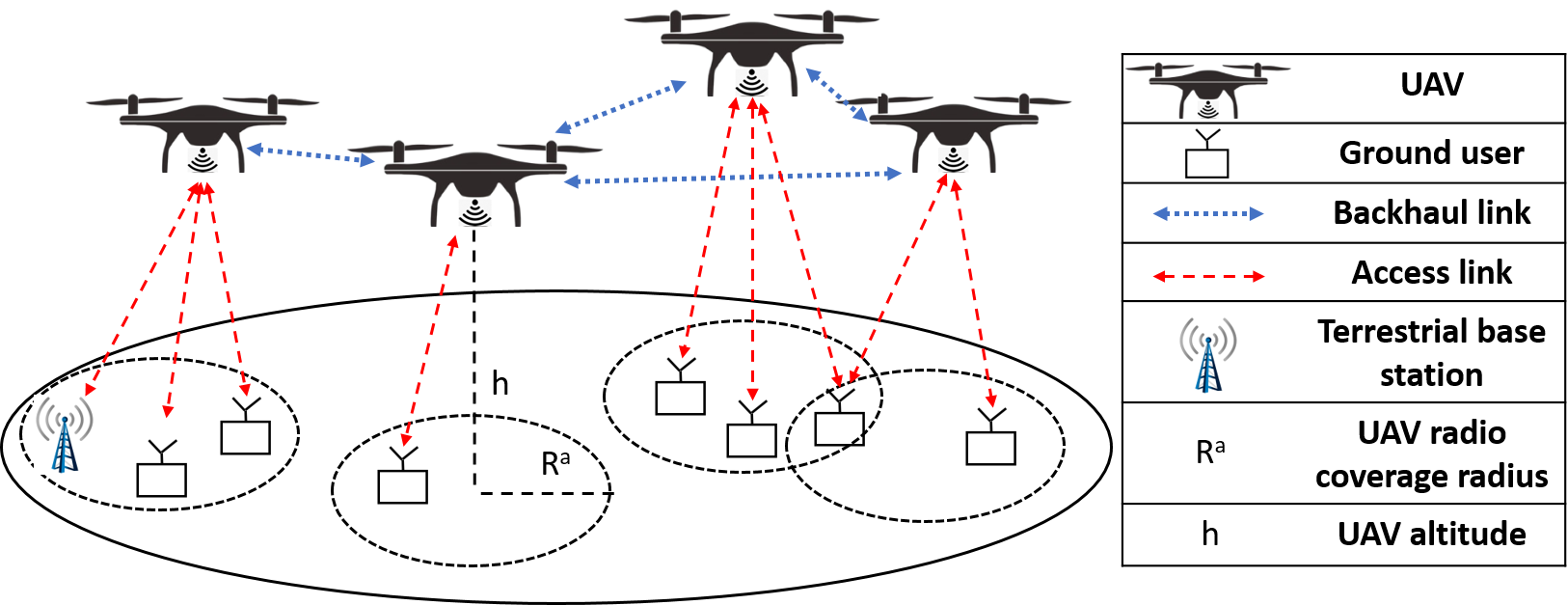}
  \caption{Envisioned Multi-UAV network model with UAVs, ground users, and terrestrial base station.}\label{fig:drone_coverage}
  \end{center}
  \vspace{-0.1in}
\end{figure}

As shown in Fig. \ref{fig:drone_coverage}, we consider a \MUS with $V$ ground users~\footnote{Since the goal of our work is to provide wireless coverage to ground users in rural areas, e.g. Internet of Things devices in farmlands, they are considered to be static or low mobility ground users. High mobility ground users is out of scope of this work and will be investigated as a part of future work.}, $J$ UAVs as aerial base stations, and one or more nearby terrestrial BS via which the \MUS can communicate with outside world. Denote $B$ as the set of terrestrial BSs in the area of operation. We assume that ground users and terrestrial BSs are located at the given locations on the horizontal plane (ground). For simplicity, let $V = V \cup B$ be the set of all ground users (including the terrestrial BSs)\footnote{As we focus on the deployment and coverage of UAV network, essentially, UAVs need to cover both the ground users and BSs (i.e., BS can be assumed as a  ``ground user'' in the UAV network.)} that are required to be covered by $J$ UAVs deployed in the region of interest.

Note that both the height and horizontal distance of a UAV towards ground users significantly impact the channel conditions of the access link. On the other hand, all ground users need to communicate with each other (or to the remote server) via the backhauls between UAVs, which requires that all UAVs keep connectivity with each other. Thus, UAVs should be deployed close enough to each other to ensure the channel condition threshold of the UAV-UAV backhaul links. As a result, both the (UAV-ground user) access link and (UAV-UAV) backhaul link are taken into account in the \MUS to provide seamless coverage for all ground users. 

In this work, we consider that the adjacent UAVs (in the backhaul links) in our Multi-UAV Network utilize orthogonal frequency channels to not interfere with one another. This is a typical frequency planning approach employed in the terrestrial cellular networks. It means that two UAVs can reuse the same frequency channel (and spectrum band) if those two UAVs not adjacent to each other. For access link, we can use the same frequency channel with various frequency/time division multiple access (and furthermore, channel resource scheduling protocols) so as to ensure no or little co-channel interference between two ground users.

Even after employing proper frequency planning and multiple frequency/time division access schemes, we understand that both these communication links are wireless and are susceptible to channel interference and other factor such as, large scale fading, small-scale fading and shadowing. However, since the focus of our work is on the deployment and coverage of Multi-UAV Network (rather than lower PHY or MAC layer design for a certain wireless link), we utilize simplified communication path loss model for wireless links (both access and backhaul links). One can  extend our work to account for channel interference and other specific channel related factors by employing sophisticated path loss models, such as, measurement-based path loss models.

\vspace{-0.15in}
\subsection{Average path loss model between a ground user and UAV}
The air-to-ground communications channel may be line-of-sight (LoS), and non-line-of-sight (NLoS), and thus both of them need to be taken into account. The communication channel is assumed to be a probabilistic LoS channel. Given an access link between ground user $i$ located at $(x'_i,y'_i,0)$, and UAV $j$ located at $(x_j,y_j,h_j)$, the path loss of the LoS and NLoS channels are modeled as follows \cite{al2014optimal}:

\begin{equation}
\begin{cases}
{\varphi_{ij}^L = 20\log \left(\frac{{4\pi {f_c}{d_{ij}}}}{c}\right) + {\eta ^L},}\\
{\varphi_{ij}^N = 20\log \left(\frac{{4\pi {f_c}{d_{ij}}}}{c}\right) + {\eta ^N}.}
\end{cases}
\end{equation}

where $\eta^L$ and $\eta^N$ are denoted as the system loss for the LoS and NLoS of an access link, $f_c$ is the carrier frequency, $d_{ij}=[(x_j-x'_i)^2+(y_j-y'_i)^2+h_j^2]^{\frac{1}{2}}$ is this 3D distance between user $i$ and UAV $j$, and $c$ is the speed of light. 
As the access link is a probabilistic LoS channel, the probabilities of both the LoS and NLoS channels can be expressed as \cite{al2014optimal}:

\begin{equation}
\left\{ {\begin{array}{*{20}{c}}
{{p^L} = {{[1 + a{e^{ - b({\theta _{ij}} - a)}}]}^{ - 1}},}\\
{{p^N} = 1 - {p^L},}
\end{array}} \right.
\end{equation}
where $a$ and $b$ are constant parameters based on the environment (e.g., urban, rural, etc.), which can be measured proactively. $\theta_{ij}$ is the elevation angle between ground user $i$ and UAV $j$ that can be modeled as $\theta_{ij}=arctan(h_j/\delta_{ij})$ where $h_j$ is the height of the UAV and $\delta_{ij}=[(x_j-x_i)^2+(y_j-y_i)^2]^{\frac{1}{2}}$ is the horizontal distance between user $i$ and UAV $j$.

Hence, the average path loss of the access link can be derived as follows \cite{al2014optimal}:
\begin{equation}
    \varphi_{ij}=p^{L}\varphi^{L}_{ij} + p^{N}\varphi^{N}_{ij}.
\end{equation}

\vspace{-0.2in}
\subsection{Average path loss model between two UAVs} Since UAVs fly over ground users, we assume that their altitudes are high enough to maintain LoS channel between each other. Thus, the path loss between UAV $j$ and UAV $k$ can be expressed as
\begin{equation} \label{pathloss_value_cal}
    {\varphi _{jk}^L = 20\log \left(\frac{{4\pi {f_c}{d_{jk}}}}{c}\right)}.
\end{equation}

Without loss of generality and for the ease of presentation, we consider same carrier frequency ($f_c$) for backhaul link (as that of access link).

\begin{table}
	\caption{List of Symbols}
	\label{tab:symbols}
	\centering
	\begin{tabular}{|p{0.6in}|p{2.5in}|}
	\hline
		\textbf{Symbol} & \textbf{Definition}  \\
		\hline
    	$B$ & Set of terrestrial BSs. \\ \hline
    	$V$ & Set of a ground users (including terrestrial BSs). \\ \hline
		$J$ & Set of UAVs in \MUS. \\ \hline
		$(x_j,y_j,h_j)$		& 	3D location of UAV $j$.\\ \hline
		$(x'_i,y'_i,0)$		& 	3D location of ground user $i$.\\ \hline
		$\gamma_{ij}$   &  SNR between UAV $j$ and user $i$\\ \hline
		$\gamma'_{jk}$	& SNR between UAV $j$ and UAV $k$. \\	\hline
		$\gamma_0$ & Minimum SNR threshold for successful connection between a ground user $i$ and a UAV. \\ \hline
		$\gamma_0^{'}$ & Minimum SNR threshold for a successful communication link between two UAVs. \\ \hline
		$R$ & Radius of maximum radio coverage on ground  \\ \hline
		$R^{'}$ & Radius of maximum UAV coverage, (i.e., in air) \\ \hline
	\end{tabular}
	\vspace{-0.2in}
\end{table}

\vspace{-0.15in}
\subsection{Communications Model} Given the path loss model, $P_j$ as the transmission power of UAV $j$, and $\delta^2$ as the noise power, we can derive the signal-to-noise ratio (SNR) between UAV $j$ and ground user $i$ as follows:

\begin{equation} \label{SNR_value_cal}
    \gamma_{ij}= 10\log{P_{j}} - \varphi _{ij} - 10\log{\delta^2}.
\end{equation}

Note that the SNR of a ground user $i$ determines if it is covered by the corresponding UAV. In other words, a ground user is assumed to be within the coverage of UAV $j$ when its $\gamma_{ij}$ meets the SNR threshold $\gamma_0$ (i.e., $\gamma_{ij} \geq \gamma_0$). Given the 3D location and transmission power of UAV $j$, we can determine the transmission range of a certain UAV. Similarly, we assume that UAV $j$ and UAV $k$ have a successful connection provided that $\gamma_{jk}$ exceeds the corresponding threshold $\gamma^{'}_{0}$.  Major notations used in this paper are listed in Table \ref{tab:symbols}.

\section{Problem Formulation} 
\label{secProblem}
In this section, we formulate the \textit{Backhaul and coverage aware  Drone deployment} (BoaRD) problem as an Integer Linear Programming (ILP) optimization problem. \Board aims to minimize the number of UAVs needed to provide wireless coverage to all ground users in the area of operation, such that (i) each ground user is efficiently covered by at least one UAV, and (ii) the backhaul connectivity is guaranteed.
 
\textbf{Notations.} Let \textit{V} and \textit{J} denote the set of ground users and UAVs, respectively. $\gamma_{ij}$ denotes the SNR between ground user $i \in V$ and UAV $j \in J$, whereas $\gamma^{'}_{jk}$ denotes the SNR between two UAVs $j$ and $k$. Let $\gamma_{0}$ denotes the minimum SNR threshold for successful connection between a ground user and a UAV. Similarly, let $\gamma^{'}_{0}$ denotes the minimum SNR threshold for successful communication between any two UAVs. For backhaul connectivity criteria, a UAV must have a link with at least one other UAV. Finally, let $(x_j, y_j, h_j)$ denotes the 3D location of UAV $j$ (in case of deployment in the 3D space) in the considered area of operation.

\textbf{Variables.} We introduce a decision variable $f_j$ = 1, if a UAV $j$ is located at a 3D location $(x_j, y_j, h_j)$ in the area of operation, otherwise $f_j = 0$. We use another variable $\beta_{ij}$ = 1, if a ground user $i \in V$ is associated with a UAV $j \in J$, otherwise $\beta_{ij} = 0$. We consider a third variable $z_{jk} = 1$, if a UAV $j$ is connected to UAV $k$, otherwise $z_{jk} = 0$ (where $j, k \in J$). Fourth, we introduce a dummy decision variable $u_{jk} = 1$ to indicate whether a certain UAV-UAV backhaul link is selected in the Multi-UAV Network topology.

\begin{align}
\vspace{-0.2in}
& \hspace{5mm} \mathop {\min}\limits_{( \beta _{ij}, f_j, z_{jk})} \sum\limits_j f_j \label{objFunction}  \\
    &\sum\limits_{j \in J} \beta _{ij}  \geq 1, \forall i \in V, \label{robustnessConst1}\\
    &\gamma _{ij} \beta _{ij} \geq   \gamma _0 \beta _{ij},  \forall i \in V, \forall j \in J, \label{userLinkConst} \\
      &\beta_{ij} \leq f_j, \label{userLinkExistenceConst}\\
    & \gamma _{jk} z_{jk} \geq  \gamma _0^{'}z_{jk},  \forall j \in J, \forall k \in J, k \neq j, \label{droneLinkConst}\\
    & z_{jk} \leq f_j, z_{jk} \le f_k, \label{droneLinkExistenceConst}\\
    & \sum\limits_{k \in J} z_{jk}  \geq 1, \forall j \in J, \label{robustnessConst2}\\
    & u_{jk} \leq z_{jk}, \label{connectedNetworkConst1}\\
    & \sum_{\substack{j, k \in J, \\ j \neq k}} u_{jk} \geq \sum_{p \in J} f_p - 1, \label{connectedNetworkConst2} \\
    & \sum_{\substack{j, k \in S, \\ j \neq k}} u_{jk} \leq |S| -1, \forall S \subseteq J, |S| \geq 1, \label{connectedNetworkConst3} \\
    & f_j, \beta_{ij}, z_{jk}, u_{jk} \in \{0, 1\}. \label{binaryVariablesP1}
\end{align}

\textbf{Objective Function.} As shown in Eq. \ref{objFunction}, the objective function is to minimize the number of UAVs (and determine the optimal 3D placement of each UAV) in Multi-UAV Network, so that each ground user is served by at least one UAV and the backhaul connectivity among UAVs is guaranteed. Note that this does not preclude the possibility that some ground users may be covered by more than one UAVs, and each UAV may have links with more than one UAVs in Multi-UAV Network. 

\textbf{Constraints.} Eq. \ref{robustnessConst1} imposes each ground user to be covered by at least one UAV in the area. Eq. \ref{userLinkConst} indicates that a ground user can only be covered by a certain UAV if the SNR threshold ($\gamma_0$) is met. Eq. \ref{userLinkExistenceConst} indicates that ground user $i$ may be associated with UAV $j$ only if UAV $j$ is deployed in the area. Eq. \ref{droneLinkConst} ensures that a UAV-UAV backhaul link exists only if it meets the SNR threshold $\gamma^{'}_0$ between them. Similarly, Eq. \ref{droneLinkExistenceConst} indicates that a connection between two UAVs $j$ and $k$ exists only if both UAVs are deployed in the area of operation. Eq. \ref{robustnessConst2} ensures that a certain UAV has a backhaul link with at least one other UAV. Note this constraint by itself does not guarantee that the backhaul connectivity constraint is met. 

In order to address this, we introduce three more constraints, i.e., Constraints \ref{connectedNetworkConst1}, \ref{connectedNetworkConst2}, and \ref{connectedNetworkConst3}. Eq. \ref{connectedNetworkConst1} restricts a certain UAV-to-UAV link be a part of connected network topology only if the communication link exists between those two UAVs. Constraint \ref{connectedNetworkConst2} ensures that there exists at least ($n - 1$) connections among $n$ UAVs, where $n = \sum_{j \in J} f_j$ and $n \leq |J|$. Moreover, constraint in Eq. \ref{connectedNetworkConst3} ensures that there is no cycle in any subset $S \subseteq J$. These two are the necessary and sufficient conditions to ensure a connected backhaul connectivity in \MUS with $n$ UAVs. Eq. \ref{binaryVariablesP1} represents the binary decision variables, that take values either $0$ or $1$.

\begin{theorem}
The \Board problem is NP-Hard.
\end{theorem}

\begin{proof}
We show that Geometric Set Cover (GSC) problem, which is NP-Hard, is polynomial-time reducible to the \Board problem, which further shows that the \Board problem is NP-Hard. Let us consider a generic instance of GSC problem: Given a set $X$ of points in $\mathbb{R}^2$, and $\mathcal{R}$ is a family of subsets of $X$, which is called ranges. The goal is to select a subset $\mathcal{C}\subseteq \mathcal{R}$ whose size is minimum and all points in the $X$ are covered by at least one range in $\mathcal{C}$.

The proof is quite straightforward. Suppose that the SNR threshold for UAV-to-UAV link is $-\infty$ (i.e., $\gamma^{'}_{0}  = -\infty$), which relaxes Constraint \ref{droneLinkConst}. Furthermore, it also means all UAVs are in communication range, and thus, relaxes all connectivity constraints, i.e., Constraints \ref{connectedNetworkConst1}, \ref{connectedNetworkConst2}, and \ref{connectedNetworkConst3}. Given SNR threshold for UAV-to-ground user link $\gamma_0$, we can compute the radius of maximum ground coverage (denoted by $R$) \cite{al2014optimal}. Then, \Board Problem is transformed into choosing the minimum number of radio coverage disks (with radius $R$) that provides coverage to all ground users. Considering the set of users $J=X$, and the set of users on the ground that can be covered by deploying UAVs $R=\mathcal{R}$, then the optimal solution to the transformed problem $C$ is equivalent to the optimal solution to the GSC problem $\mathcal{C}$.
We reduced the GSC problem to an instance of \Board problem. Since the GSC problem is a classic NP-hard problem~\cite{BRIMKOV20121039}, the \Board problem is also NP-Hard.

\vspace{-0.15in}
\end{proof}

\section{Key Intuitions Behind the Proposed Solution}
\label{secPreliminaries}

In this section, we discuss key intuitions behind the proposed solution (discussed in Section VI) that solves \Board problem in polynomial-time with performance guarantees.

\textit{(I) Focus on the optimal altitude of the UAVs that provides maximum radio coverage on the ground.} Since one of the key objectives of \Board problem is to provide coverage to all ground users in the area, it becomes intuitive that UAVs are deployed at the optimal altitude that provides maximum radio coverage on the ground while meeting the SNR threshold, thanks to their capability of maintaining that height. From the seminal work \cite{al2014optimal}, such an optimal altitude can be computed mathematically given the environmental parameters, carrier frequency, LoS and NLoS system loss, and transmit power. Let $h$ and $R$ respectively be the optimal altitude of UAV and the radius of the maximum radio coverage by a UAV on the ground (corresponding to SNR threshold $\gamma_0$ for access link).

\textit{(II) Focus on possible covered sets of ground users for UAVs rather than candidate 3D locations of UAVs.} Since UAVs can be deployed in any 3D location in the area of operation $\mathbb{R}^3$, the number of candidate locations of UAVs is infinite, or the solution space of \Board problem is infinite. However, several of those candidate locations are essentially equivalent if they provide wireless coverage to the same set of ground users. Thus, the Proposed solution only needs to consider one representative 3D location among its associated class of all equivalent 3D locations, and the number of all such representative 3D locations is finite because the number of all possible covered set of ground users is finite. 

\textit{(III) Focus on those candidate 3D locations that provide coverage to a large set of ground users.} If a representative 3D location covers the set of ground users $\{u_1, u_2, u_3, u_4 \}$, then the Proposed solution does not need to consider 3D locations that cover its subsets, such as, $\{u_1, u_2\}$ or $\{u_2, u_3, u_4\}$~\footnote{Note that this intuition holds for our \Board problem setting as it only ensures ground user coverage (with no constraint on bit rate thresholds). However, it will not hold for \MUS system where user bit rate thresholds have to be met. We will investigate this in our future work.}.

\begin{figure}
\begin{center}
  \includegraphics[scale=0.75]{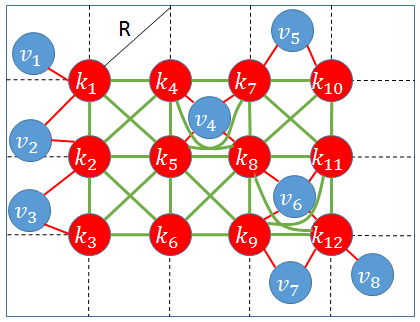}
  \caption{Constructed graph $G$, with $V$ ground users and $\mathcal{K}$ candidate locations in the area of operation. Blue and Red dots respectively represent ground users and candidate 3D locations. An edge (in red) represents the communication link between a ground user and a candidate location. Whereas, an edge (in green) denotes the link between two UAVs. }\label{fig:graph_constr}
  \end{center}
  \vspace{-0.25in}
\end{figure}

\vspace{-0.05in}

\textit{(IV) Account for UAV-UAV connectivity guarantee.} The proposed solution to the \Board problem not only minimizes the number of UAVs needed to cover all the ground users, but also guarantees that the deployed UAVs ensure backhaul connectivity. Now the average path loss in the UAV-UAV backhaul link is usually much smaller compared to that of UAV-to-ground user access link, thanks to LOS path loss model in case of backhaul links. Thus, the maximum radius $R^{'}$ of UAV coverage is much larger than the maximum radius $R$ of the radio coverage on the ground, i.e., $R^{'} > R$.

\textbf{3D candidate locations of UAVs.}
Intuitively, the number of candidate locations of UAVs in the area of operation is infinite. Utilizing intuitions I and II, we present a simple grid approach to reduce the infinite solution space of \Board problem to finite solution space -- We first determine the optimal altitude of the UAVs that provides the maximum coverage on the ground (See Intuition I). Next, as depicted in Figure \ref{fig:graph_constr}, we divide the area of operation into $N$ grids, each with diagonal length = $R$, and length = breadth = $\frac{R}{\sqrt{2}}$, where $R$ is the radius of the maximum radio coverage on the ground. All these grid intersection points are the candidate locations of UAVs. Notice this approach allows $8$ grid intersections ($\{k_i\mid i\in[1,9],i\neq5\}$) in addition to the central grid intersection ($k_5$). Given that there are $9$ candidate locations within a grid of diagonal length $R$, it is very likely that the resultant finite solution space (with $N$ grids) will have the (near) optimal solution to the problem. Compared to state-of-the-art grid based approach~\cite{kalantari2017backhaul, mayor_deploying_2019} that simply discretizes the area of operation, our approach here is to create a minimum number of grids that discretizes the area without loosing out on the potential of finding an optimal solution. However, our proposed approach is independent of this approach, and will work perfectly with any state-of-the-art grid discretion approaches.

\section{Proposed Algorithm}
\label{secSolution}

In this section, we model \MUS as a graph (See Figure \ref{fig:graph_constr}), and then transform the \Board problem as a graph problem. Next, we propose a low-complexity algorithm with performance guarantee to effectively solve the transformed \Board problem. 

\textbf{Graph modeling.} A UAV will provide wireless coverage to a certain ground user if it lies within the UAV's maximum radio coverage ($R$). On the other hand, a UAV will have a backhaul link with another UAV in the area of operation, if they are within the transmission range of each other (i.e., $R^{'})$. As depicted in Figure \ref{fig:graph_constr}, let graph $G(V \cup \mathcal{K}, E \cup E^{'})$ denotes the \MUS with $V$ as the set of ground users and $\mathcal{K}$ as the set of 3D candidate locations. For edge set $E$, we include an edge $e_{ik}$ between user $i \in V$ and candidate location $k \in \mathcal{K}$, if user $i$ is covered by the UAV at location $k$, i.e., $\gamma_{ik} \geq \gamma_0$. In other words, the euclidean distance between user $i$ and candidate location $k$ is less than or equal to $R$, i.e., $dist_{ik} \leq R$. Similarly, the edge set $E^{'}$ reflects the backhaul links, where a link $e^{'}_{jk}$ between UAV $j$ and $k$ exists if they are within each other's transmission range. i.e., $\gamma_{jk} \geq \gamma^{'}_0$ (or euclidean distance between UAV $j$ and $k$, $dist_{jk} \leq R^{'}$). 

\textit{Given the graph $G$, the \Board problem is transformed into choosing the minimum subset $D$ of $\mathcal{K}$ such that every node in $V$ is adjacent to at least one member of $D$ and the members in $D$ form a connected network topology.}

Let $D \subseteq \mathcal{K}$ be the minimum connected subset of $G$, i.e., the solution of the \Board problem. It means that all nodes in $D$ can cover every node in $V$, and keep connected (i.e., the nodes of $D$ can reach each other via a path that stays entirely within $D$). In other words, deploying $|D|$ UAVs at the selected locations can provision coverage to all ground users, whereas all the UAVs form a connected network, i.e., backhaul connectivity is guaranteed. Note that though the \Board problem seems similar with a well-known graph problem, i.e., \textit{Minimum Connected Dominating Set} (MCDS) problem{\footnote{A dominating set of a graph $G$, denoted as $D$, is a set of nodes where for every node $u\in G$, either $u\in D$ or $u$ is adjacent to a node $v\in D$. Then, A connected dominating set of a graph $G$, denoted as $D$, is a dominating set of graph $G$ where for every node in $D$ there exists a path to any other node in $D$ that stays entirely within $D$. And finally, A minimum connected dominating set (MCDS) of a graph G is a connected dominating set with the smallest possible cardinality among all connected dominating sets of G.}}~\cite{Butenko2004}, it is different in two aspects: (i) $D$ is a subset of $\mathcal{K}$ (instead of being a subset of $V$ itself); and (ii) each node in $V$ is adjacent to at least one node in $D$, where $D \subseteq \mathcal{K}$. Note, we do not consider any connectivity constraint on remaining $\mathcal{K} \setminus D$. Next we detail the Proposed algorithm that efficiently solves the transformed \Board problem with low computational complexity and provable performance guarantee. 

\begin{algorithm}  
\label{alg_2}
\small
    \textbf{Input:} Locations of ground users, Radius of maximum radio coverage on ground $R$, and Radius of maximum UAV coverage $R^{'}$
	\textbf{Output:} Graph $G$ and List of candidate locations with its associated ground users, $U$
	\begin{algorithmic} [1]
    \State Initialize a graph $G = \phi$, and a list $U = \phi$
    \For{candidate location, $j \in \mathcal{K}$}
        \For{candidate location, $k  \in \mathcal{K}$}
        \If{$dist_{jk} \leq R^{'}$}
            \State $G$.add\_edge($j, k$)
        \EndIf
        \EndFor
    \EndFor
     \State $U_k = \phi$ //Set of ground users in coverage range of candidate location $k \in \mathcal{K}$, i.e., $dist_{ik} \leq R$ where $i \in V$
    \For{ground user, $i \in V$}
        \For{candidate location, $k  \in \mathcal{K}$}
        \If{$dist_{ik} \leq R$}
             \State $G$.add\_edge$(i, k)$
            \State $U_k = U_k \cup i$
        \EndIf
        \EndFor
    \EndFor
    \State $U = \cup_{k \in \mathcal{K}} V_k$
    \State return $G, U$
    \end{algorithmic}  
	\caption{Initialization}
	\label{P1_algo}  
\end{algorithm}	

\begin{algorithm}  
\small
	\textbf{Input:} $V$ ground users, $\mathcal{K}$ candidate locations, and Maximum ground and UAV coverage radius, $R$ and $R^{'}$ respectively.\\
	\textbf{Output:} $D$ chosen candidate locations, where $D \subseteq \mathcal{K}$
	\begin{algorithmic} [1]
	\State G, U = INITIALIZATION (V, $\mathcal{K}$, $R$, $R^{'}$)
    \State $D = \mathcal{K}$ 
    \State $F = \phi$ //Set of fixed nodes
    \While{$D \setminus F \neq \phi$}
    \State $U_{min} = \{k\mid U_k\in U,|U_k|$ is minimum$\}$ 
    \State $u = argmin\{\delta(k)\mid k \in U_{min} \}$
    \If{$G[D \setminus \{u\}]$ is not connected or $V \setminus \cup_{k \in D \setminus \{u\}} U_k \neq \phi$}
        \State $F = F \cup \{u\}$
        \For{candidate location, $k  \in D \setminus \{F\}$}
            \For{ground user, $i \in U_u$}
                \State $U_k = U_k \setminus \{i\}$
            \EndFor
        \EndFor
    \Else
        \State $D = D \setminus \{u\}$
    \EndIf
    \State $U = U \setminus \{U_u\}$
    \EndWhile
    \State return $D$
	\end{algorithmic}  
	\caption{Proposed Algorithm}
	\label{alg:BoaRD}  
\end{algorithm}

\textbf{Algorithm Description.} The pseudocode of the Proposed algorithm is presented in Algorithm 2. As shown in Line 1, the algorithm first calls a function INITIALIZATION ($V$, $\mathcal{K}$, $R$, $R^{'}$) that returns -- (i) the aforestated graph $G$, and (ii) a list $U = \cup_{k \in \mathcal{K}} U_k$ where each $U_k$ contains the set of the ground users associated with the UAV at candidate location $k \in \mathcal{K}$. 

The details of INITIALIZATION ($V$, $\mathcal{K}$, $R$, $R^{'}$) is presented in Algorithm \ref{P1_algo}. At the beginning, the algorithm initializes an empty graph $G$ and list $U$ in Line 1.  Then, as shown lines 2 - 5, the graph $G$ incorporates edges between each pair of candidate location $j, k \in \mathcal{K}$, whose euclidean distance is $\leq R^{'}$ (in other words, the UAVs at the candidate locations $j$ and $k$ are within the maximum UAV coverage radius $R^{'}$ of each other). Similarly, the algorithm also adds edges between ground user $i \in V$ and candidate location $k \in \mathcal{K}$ if the euclidean distance between them is $\leq R$ (See lines 7-10). In addition, as shown in Line 11, we add ground user $i$ into $U_k$ (w.r.t to candidate location $k$) if user $i$ is within the coverage range of UAV $k$.

After the initialization step, we take the set of candidate locations $\mathcal{K}$ as the initial $D$ (Line 2 of Algorithm 2). (Note $D$ is the minimum connected subset of $G$, i.e., the solution to the \Board problem). At each iteration, the algorithm first selects the list of candidate locations $U_{min}$ which have the minimum cardinality of covered ground users (See line 5). Sequentially in Line 5, it selects a candidate location $u$, which has the minimum degree in set $U_{min}$. As shown in Lines 7 - 14, the algorithm removes the candidate location $u$ from graph $G$ if both of the following conditions are met -- (1) Removing $u$ does not make the graph $G$ disconnected, and (2) Remaining candidate locations ($K \setminus \{u\}$) can cover all the nodes in ground user set $V$. If the above conditions are not satisfied, candidate location $u$ is dispensible in the final solution $D$, and thus, $u$ must be \textit{fixed} (See Line 8). Following this in Lines 9 - 11, all the ground users that are associated with $u$ are removed from the remaining $U_k$. Otherwise as shown in line 13, we remove the candidate location $u$ from $D$. Such a candidate location is referred to as \textit{non-fixed candidate location} as its removal neither disconnects the subgraph in $D$ nor hampers the coverage of ground users in $V$. Afterwards, we remove the corresponding set $U_u$ from the list $U$. These steps are repeated until there is no non-fixed candidate location in $D$. 

\begin{theorem}
The time complexity of the Proposed algorithm is $O(|\mathcal{K}|^2.(|V|+|\mathcal{K}|)$.
\end{theorem}

\begin{proof}
Since the Proposed algorithm calls the INITIALIZATION function, we first check the time complexity of Algorithm \ref{P1_algo}. The time complexity for lines 2-5, and 7-11 of Algorithm \ref{P1_algo} are $O(|\mathcal{K}|^2)$ and $O(|V|.|\mathcal{K}|)$, respectively. Hence, the running time of Algorithm \ref{P1_algo} (or line 1) of the Proposed algorithm becomes $O(|K|.(|V|+|K|))$. Line 5, and 6 of the Proposed algorithm have the time complexity $O(|\mathcal{K}|.|log(|\mathcal{K}|))$ (the time needed to sort the list and make a binary search), and $O(|\mathcal{K}|^{2})$ respectively. In addition, the procedure of checking if a graph $G$ (with $D$ nodes and $E$ edges) is connected or not, has the time complexity $O(|D|+|E|)$, which is the time needed for running the depth first search. It implies that the running time of the first condition in line 7 is $(O(|\mathcal{K}|+|\mathcal{K}|^2))$. The second condition in line 7 has the complexity of $O(|V|.|\mathcal{K}|)$. The cost of lines 9-11 is $O(|\mathcal{K}|^{2})$. The \textbf{while} loop in line 4 will repeated for $|\mathcal{K}|$ time, because in each step a node in $\mathcal{K}$ will be either fixed or removed. Hence, the total time complexity of the Proposed algorithm can be expressed as $O(|\mathcal{K}|(|V|.|\mathcal{K}|+|\mathcal{K}|^2))$ = $O(|\mathcal{K}^2|(|V| +|\mathcal{K}|))$.
\end{proof}

\vspace{-0.1in}

\section{Performance Guarantee Analysis}
\label{secPerformanceProof}

In this section, we analyze the performance guarantee of the Proposed algorithm presented in aforestated Section \ref{secSolution}. In order to accomplish this, we first reduce graph $G$ into graph $G_r$ in the following manner. 

Consider a candidate location $i \in D$ that provide wireless coverage to ground users which are not in the coverage area of any $j \in D \setminus \{i\}$. Let $U_i \in V$ be the set of ground users in the coverage area of candidate location $i$. We remove all ground users in the set $U_i$, and instead, add one dummy node $u_i$ (representing all the ground users in $U_i$)  connected to candidate location $i$. By repeating this procedure, until no ground users remain, we will have a set $\mathcal{U}= \cup_{i} u_i$. Let graph $G_r(\mathcal{U} \cup \mathcal{K},E'\cup E'')$ where $E''$ is the set of edges between added nodes $u_i \in \mathcal{U}$ and UAV $i$. Note that each node in $D$ is connected to a unique node $u_i \in \mathcal{U}$ or employed to ensure the connectivity of subset $D$, otherwise it could have been removed (recall such a node corresponds to a fixed candidate location in Algorithm \ref{alg:BoaRD}).

\begin{lemma} \label{max_IS}
Let OPT be an MCDS of $G_r$, any maximal independent set{\footnote{Set $S$ is an independent set if the subgraph $S$ contains no edges. Then, an independent set $S\subset G$ is a maximal independent set (MIS) if and only if for every vertex $u\in G-S$ the set $S\cup\{u\}$ is not independent (i.e. the graph $S\cup\{u\}$ is not independent).}} of $G_r$ has a maximum size of $5|OPT|+1$.
\end{lemma}

\begin{proof}
Inspired by the work in \cite{alzoubi_distributed_2002}, assume $OPT$ be any MCDS in $G_r$, and $W$ is any MIS of $G_r$. Let $T$ be an arbitrary spanning tree of $OPT$, and $\{v_1, v_2,\dots, v_{opt}\}$ be any arbitrary preorder traversal of $T$ after picking an arbitrary node as the root of $T$. Let $W_i$ be the set of vertices in $W$ that are adjacent to $v_i$, but none of $v_1, v_2, …, v_{i-1}$, for any $1 < i \leq opt$. Also, let $W_1$ be the set of vertices in $W$ that are adjacent to node $v_1$. Then, $W_1, W_2,\dots, W_{opt}$ form a partition of $W$. $W_1$ has at most one of the $u_i$s defined in the above discussion. All other vertices of the $W_1$, which are nodes from the set $K$, should have a pairwise distance of more than $R’$ and thus there could be at most five nodes. Therefore, $|W_1|$ could be at most 6. In other words, let us scale down all distances from $R’$ to 1, and assume two UAVs will be connected if their euclidean distance is $\leq 1$. Hence, any node will be adjacent to a maximum of five independent nodes in a Unit Disk Graph~\cite{alzoubi_distributed_2002}. Because of the preorder traversal labeling of the $T$, there is at least one node in $\{v_1, v_2, …, v_{i-1}\}$ adjacent to $v_i$, and let us name it $v_j$. Hence, by considering the coverage range of $v_i$ and $v_j$, there is a sector of at most 240 degrees, which is within the coverage area of $v_i$, but not $v_j$. Therefore, $W_i\setminus\{u_i\}$ should lie in the mentioned sector, and thus $|W_i\setminus\{u_i\}|$ is at most 4, and $|W_i|$ is at most 5. Therefore, \\$|W|=\sum_{i=1}^{opt} |W_i| \leq 6+5(opt-1) = 5.opt+1$.
\end{proof}

\begin{lemma}
Given the graph $G_r$ and the resulting subset $D$ calculated by Proposed Algorithm on graph $G$, there is an independent set of $G_r$ containing at least $|D|/4$ vertices.
\end{lemma}

\begin{proof}
The goal is to find an independent set with the cardinality of at least $\lceil|C_m|/2\rceil$ for each induced cycle{\footnote{An induced cycle is a cycle such that no two nodes of the cycle are connected by an edge that does not itself belong to the cycle.}} $C_m=\{v_1, v_2, …, v_k\}$ in graph $G[D]$, and remove them afterwards \cite{Butenko2004}. If $|C_m|$ is even, then there is a set $I_m=\{v_i\mid i$ is even$\}$, which, by the definition of an induced cycle, is an independent set. However, if $|C_m|$ is odd, then there are two cases. First, at least one of the vertices in $C_m$, let us name it $v_i$, have a neighbor $u_i$, which is not the neighbour of other vertices in $D$, by our definition. Therefore, we can construct $I_m$ by picking $u_i$ with half of nodes from $C_m\setminus\{v_i\}$, which is a bipartite tree, so that it is independent. Second, if none of vertices in $C_m$ is connected to a user-representation node $u_i$, it implies that all the nodes in $C_m$ are fixed to ensure the connectivity of the graph, otherwise at least one of them would have been removed by Algorithm \ref{alg:BoaRD}. In this case, assume $v_i$ is a vertex of $C_i$, and it is connected to subgraph $L_i$, which is out of $C_m$. It is obvious that if we remove $v_i$ from $C_m$, there is no path from any vertices in $L_i$ to any vertices in $C_m$. Therefore, each vertex $v_i \in C_m$ is connected to a unique subgraph $L_i$. There is at least one user-representation node $u_j$ in each $L_i$, otherwise all vertices of $L_i$ should have been removed. Let us construct $I_m$ by selecting $u_j$ with half of nodes from $C\setminus\{v_i\}$ so that it is independent. We should ensure that if we select $u_j$ to construct the independent set $I_m$ corresponding to set $C_m$, $u_j$ is not needed for constructing $I_l$ corresponding to set $C_l$. If $L_i$, itself, consists of any induced cycle(s) $C_l$ with odd number of vertices, then we can find a node $u_k$ for constructing $I_l$. Because $|C_l|>2$, then $C_l$ is connected to subgraphs $\{L^{'}_1,...,L^{'}_{|C_l|-1}\} \subset L_i$. It implies that there are more than one user-representation node to select for constructing $I_m$, and $I_l$. It is clear that when the number of vertices in the graph $D$ is finite, then the number of such induced cycles is finite. 

After removal of all induced cycles from $D$, independent trees will compose the remaining set. As shown in \cite{Butenko2004}, we can select at least half of the nodes in each tree so that they compose independent sets. An independent set $I^{*}$ can be constructed from the union of obtained independent sets with the cardinality of at least $1/2\times|\cup_{i}I_i|$, and thus $|I^{*}| \geq |D|/4$.
\end{proof}
Similar to \cite{Butenko2004}, theorem provided as follows can estimate the Proposed algorithm performance guarantees.
\begin{theorem}
The cardinality of the subset D computed by the Proposed algorithm is at most $20OPT+4$, where $OPT$ is the size of an MCDS of $G_r$.
\end{theorem} 
\begin{proof}
From the Lemma 2 and 1, we have

$5.opt+1 \geq |Maximal\:Independent\:Set| \geq |I^{*}| \geq |D|/4 \\
 \Rightarrow 20.opt+4 \geq |D|$.
\end{proof}

\vspace{-0.2in}
\section{Simulation Experiments}
\label{secResults}
In this section, we first showcase how well the Proposed algorithm performs compared to the Optimal algorithm (using ILP solver) for small scenarios only (ILP solver takes several hours, rather days, to run for larger scenarios.). Following this, we focus on the performance analysis of the Proposed algorithm against three comparison/baseline algorithms, namely, Greedy algorithm (No backhaul constraint), Backhaul-aware Greedy (BaG) algorithm and Random algorithm, for general (larger) scenarios in the rest of the section.

\vspace{-0.15in}

\subsection{Comparative Analysis of Proposed and Optimal algorithms}

Before we evaluate the performance of the proposed algorithm for general simulation scenarios, we first showcase how our Proposed algorithm compares with the Optimal algorithm (solved using ILP solver called Gurobi optimizer). The experiments are done for a small scenario where the size of the area is $9 \times 9$ sq. Km, and the number of users is 40. The rest of the simulation parameters are listed in table \ref{tab:simulation_parameters}. 

As shown in table \ref{Table:optimal_vs_Proposed}, the Proposed algorithm utilizes the same or slightly higher number of UAVs for forming a multi-UAV networks compared to that of Optimal approach. Figure \ref{fig:resulting_graphs_proposed_vs_optimal} shows the resultant Multi-UAV network using Proposed and Optimal algorithms. These results show that the Proposed algorithm performs really well and is close to the Optimal algorithm for smaller scenarios. 

Moreover, table \ref{Table:exec_time_optimal_vs_Proposed} shows the exponential time complexity of the Optimal algorithm for increasing area size (i.e., increasing number
of candidate locations). We can see that running Optimal algorithm may take several hours for larger scenarios, and thus, we do not show the results for Optimal algorithm in the large-scale simulation experiments presented in the following section.

\vspace{-0.1in}

\begin{table}[ht!]
\caption{No. of UAVs (Optimal vs. Proposed algorithm)}\label{Table:optimal_vs_Proposed}
\resizebox{\columnwidth}{!}{
\begin{tabular}{|c|c|c|}
\hline
\textbf{Number of ground users (\#GU)} & \textbf{Approach} & \textbf{Number of UAVs}
\\ \hline
\multirow{2}{6em}{{\#GU = 20}} 
& {Optimal} & {3} \\
\cline{2-3}
& {Proposed} & {3} \\
\hline

\multirow{2}{6em}{{\#GU = 40}} 
& {Optimal} & {4} \\
\cline{2-3}
& {Proposed} & {5} \\
\hline

\multirow{2}{6em}{{\#GU = 60}} 
& {Optimal} & {5} \\
\cline{2-3}
& {Proposed} & {6} \\
\hline
\end{tabular}
}
\end{table}

\vspace{-0.1in}

\begin{table}[ht!]
\caption{Execution time (Optimal vs. Proposed algorithm)}\label{Table:exec_time_optimal_vs_Proposed}
\resizebox{\columnwidth}{!}{
\begin{tabular}{|c|c|c|}
\hline
\textbf{Area size (sq. Km)} & \textbf{Approach} & \textbf{Time for execution (s)}
\\ \hline
\multirow{2}{3em}{{8 $\times$ 8}} 
& {Optimal} & {36.67} \\
\cline{2-3}
& {Proposed} & {0.13} \\
\hline

\multirow{2}{3em}{{9 $\times$ 9}} 
& {Optimal} & {1198.89} \\
\cline{2-3}
& {Proposed} & {0.14} \\
\hline

\multirow{2}{4em}{{10 $\times$ 10}} 
& {Optimal} & {22625.85} \\
\cline{2-3}
& {Proposed} & {0.16} \\
\hline
\end{tabular}
}
\end{table}

\begin{figure}[h!]
\begin{center}
  \includegraphics[scale=0.46]{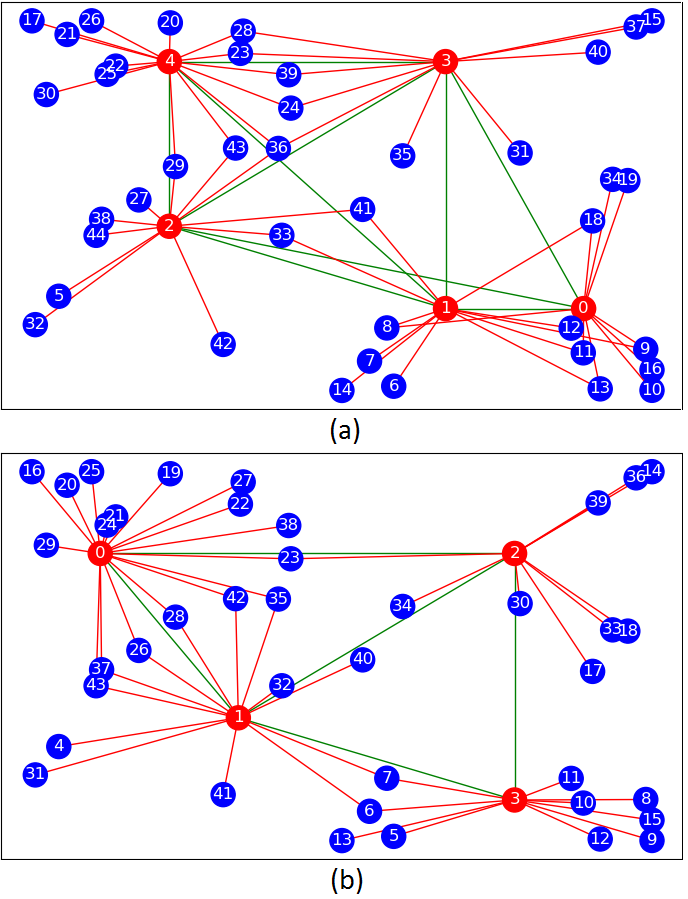}
  \vspace{-0.1in}
  \caption{Resultant Multi-UAV Network topology: (a) Proposed and (b) Optimal algorithms. Red nodes are UAV candidate locations, while blue ones are ground users. Green and red edges respectively represent the backhaul and access links.
  }\label{fig:resulting_graphs_proposed_vs_optimal}
  \end{center}
  \vspace{-0.15in}
\end{figure}

\subsection{Simulation Setting}
Unless otherwise stated, we consider a simulation area of size $50$ km $\times$ $50$ km with $200$ ground users. The major parameters of the simulation setting are listed in Table~\ref{tab:simulation_parameters}.

\begin{table}[h]
 \centering
 \caption{Simulation parameters}
 \label{tab:simulation_parameters}
    \begin{tabular}{|c|c|}
    \hline
    \textbf{Parameters}  & \textbf{Value} \\ \hline
    Simulation Area & 50km $*$ 50km \\ \hline
    Number of Users & 200 \\ \hline
    Carrier frequency ($f_c$) & 2 GHz \\ \hline
    LOS System Loss ($\eta ^L$)  & 0.1 dB \\ \hline
    NLOS System Loss ($\eta ^{N}$)  & 21 dB \\ \hline
    Environment Parameters ($a, b$) & 4.88, 0.429 \\ \hline
    UAV Transmit Power ($P_j$) & 1 W \\ \hline
    Bandwidth & 15 MHz \\ \hline
    Noise Power Spectral Density & -174 dBm/Hz  \\ \hline
    SNR Threshold for user coverage ($\gamma_{0}$) & 4 dB     \\ \hline
 \end{tabular}
\end{table}

The optimal altitude ($h$) for placing the UAVs, and its corresponding radio coverage radius on ground ($R$) are computed as $1500$ and $3300$ m, respectively for a sub-urban area as per the seminal work \cite{al2014optimal}. Refer to Table \ref{tab:simulation_parameters} for the considered the values of environmental parameters, LoS and NLoS system loss and other required parameters for calculating $h$ and $R$. Using the Eq. \ref{pathloss_value_cal}, and \ref{SNR_value_cal}, the UAV coverage radius $R^{'}$ is calculated. For instance, we observe that $R^{'} \approx 2.5*R \approx 8.3$ km when the backhaul SNR threshold, $\gamma^{'}_0 = 15$ dB. However, note that different $\gamma^{'}_0$ will result in different $R'$, and is calculated accordingly for each experiment. Also, since UAVs are hovering over the air constantly, we considered a higher SNR threshold to address the antenna pointing variations and improve the backhaul links.

In our experiments, grounds users are randomly distributed in clusters, where each cluster houses $10 - 15$ ground users, in the considered simulation area. This is a realistic user distribution in case of post-disaster scenarios~\cite{shah2017ctr} or rural/remote areas~\cite{shah2019x}, instead of completely random distribution of ground users in the entire simulation area usually considered in the UAV literature \cite{lyu2017placement}. In order to ensure the accuracy of the results, we execute each experiment $100$ times for each algorithm and take the average value as the simulation results. 

For extensive analysis, we evaluate the Proposed algorithm against the other three comparison approaches for (1) \textit{varying number of ground users}, ranging from $50$ - $500$, (2) \textit{varying area sizes}, from $(10$ km $\times 10$ km) to $(100$ km $\times 100$ km), both for varying backhaul SNR thresholds, from $10$ dB to $20$ dB.

\begin{figure*}
\begin{center}
  \includegraphics[scale=0.5]{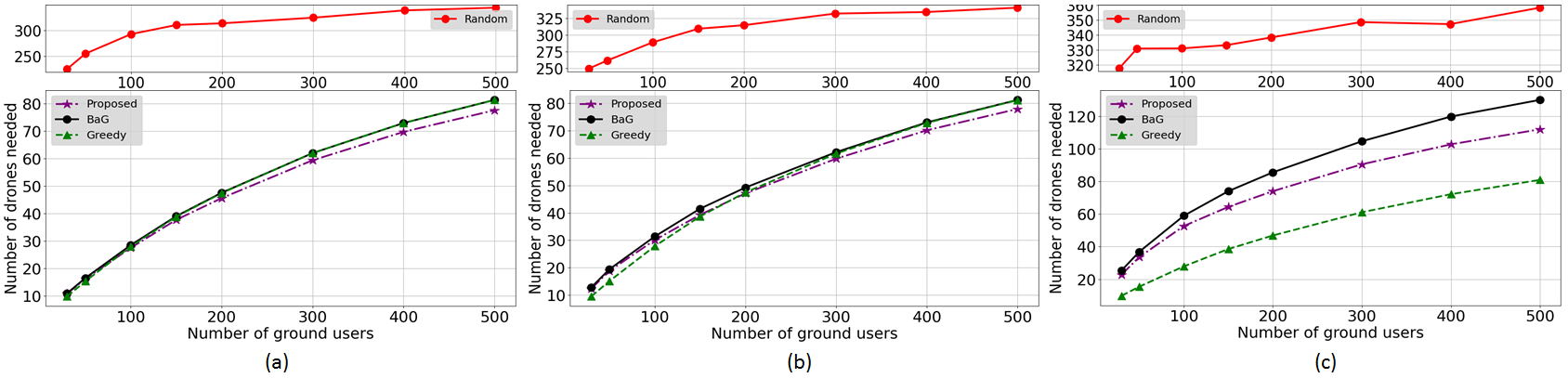}
   \caption{Number of UAVs needed vs. Number of ground users when SNR threshold $\gamma^{'}_0$ is (a) $10dB$, (b) $15dB$, (c) $20dB$
  }\label{drones_vs_users}
  \end{center}
  \vspace{-0.15in}
\end{figure*}

\begin{figure*}
\begin{center}
  \includegraphics[scale=0.5]{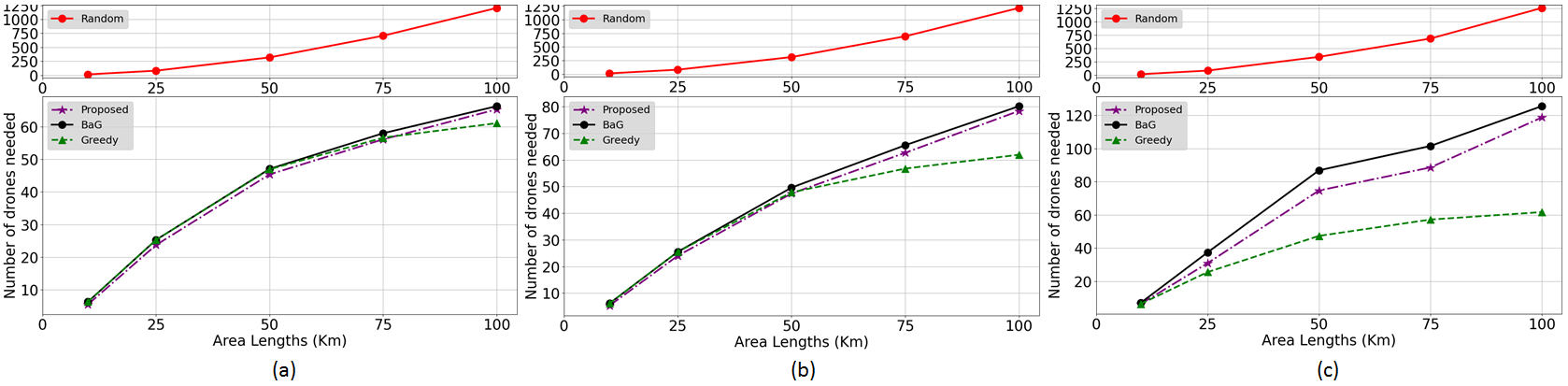}
  \vspace{-0.08in}
  \caption{Number of UAVs needed vs. Area sizes, when SNR threshold $\gamma^{'}_0$ is (a) $10dB$. (b) $15dB$. (c) $20dB$}
  \label{drones_vs_areas}
  \end{center}
  \vspace{-0.25in}
\end{figure*}

\subsection{Comparison Algorithms} \label{comparison_algorithms}

\subsubsection{Greedy Algorithm (No backhaul)} Greedy algorithm has been proposed in the literature~\cite{zhang20183D,fan2019towards} to determine the minimum number of UAVs (and its 3D placement) required to provide wireless coverage to all ground users. The algorithm works in the following manner: The algorithm first sorts the candidate locations $\mathcal{K}$ in the non-increasing order of the number of ground users associated with different locations $k \in \mathcal{K}$, i.e., ground user density. Denote $\hat{\mathcal{K}}$ as the sorted list of candidate locations. Then, the UAVs are placed sequentially at candidate locations with higher user density until all ground users are covered. Note this algorithm does not always guarantee backhaul connectivity among UAVs.

\subsubsection{Backhaul-aware Greedy Algorithm (BaG)} We extend Greedy algorithm to account for backhaul connectivity among UAVs, which we call \textit{Backhaul-aware Greedy Algorithm} (BaG). The algorithm works as follows: 

First, BaG utilizes Greedy Algorithm to deploy the minimum number of UAVs at optimal 3D candidate locations (denoted as the subset $D$) that ensures wireless coverage to all ground users. Following this, we utilize the concept of \textit{Euclidean Minimum Spanning Tree} (EMST){\footnote{Considering a set of $k$ points in the plane (or more generally in $\mathbb{R}^{d}$), its Euclidean minimum spanning tree (EMST) is a minimum spanning tree, where the weight of the edge between each pair of points is the Euclidean distance between those two points.}} to determine the minimum spanning tree of subset $D$ where the weight of the edge between each pair of UAVs is the euclidean distance between them. However, since the only possible 3D locations for placement of UAVs are candidate locations in set $\mathcal{K}$, the edge set $H$ in EMST can not be directly utilized for the placement of additional UAVs in order to ensure the backhaul connectivity. To address this, we set the weights of the existing edges between candidate locations as $1$. Now for each selected edge $e_{jk} \in H$ (where $j, k \in D$), we utilize Dijkstra's algorithm to find the shortest path between end nodes $j$ and $k$. It returns the minimum number of additional candidate locations (relay nodes), denoted by $\hat{D}$, that are required to connect the end nodes of the selected edge $e_{jk}$. We repeat the process for all edges in set $H$ and include the additional candidate locations to the set $\hat{D}$. Since there may be overlaps of candidate locations between two or more shortest paths for different edges in set $H$, we only include the additional candidate locations for a certain path if those additional candidate locations are not already included in the set $\hat{D}$. The set $D \cup \hat{D}$ is the final solution of BaG algorithm.

\subsubsection{Random Algorithm} 
Random algorithm deploys UAVs at randomly chosen candidate locations until and unless all ground users are covered and the deployed UAVs ensure backhaul connectivity in the formed Multi-UAV Network.

\subsection{Experimental Results}
\textbf{Varying number of ground users.} Fig. \ref{drones_vs_users} shows the impact of varying number of ground users on the number of required UAVs, under different backhaul SNR thresholds $\gamma^{'}_0$. In particular, we consider $\gamma^{'}_0$ as $10dB$, $15dB$, and $20dB$ and  the corresponding results for each case is reported in Fig. \ref{drones_vs_users} (a), \ref{drones_vs_users} (b), and \ref{drones_vs_users} (c) respectively. The number of UAVs required by all algorithms gradually increases with increasing number of ground users, under any considered value of $\gamma_{0}^{'}$. This is because more UAVs will be required to provide wireless coverage to increasing number of users (spread out in multiple clusters in the area). Moreover, the number of UAVs required by all algorithms (except for Greedy) increases with increasing $\gamma^{'}$, even for a fixed number of ground users. This is because UAV coverage radius $R^{'}$ decreases with increasing $\gamma_0^{'}$ and therefore, more number of UAVs would be required for ensuring backhaul connectivity. Since Greedy does not ensure backhaul connectivity, the number of required UAVs in this case remains constant and does not change with changing $\gamma^{'}_0$.

The Proposed algorithm outperforms the other two backhaul-aware algorithms, i.e., BaG and Random, for varying number of ground users under all considered backhaul SNR thresholds. The Proposed algorithm requires fewer number of UAVs by up to $15\%$ and $95\%$ compared to that of BaG and Random algorithm respectively. This shows the superiority of the Proposed algorithm in solving the BoaRD problem, i.e., optimizing the number and placement of UAVs while ensuring backhaul connectivity and ground user coverage. It is noteworthy that the difference in the number of UAVs required by the Proposed algorithm and BaG is even larger for higher $\gamma^{'}_0$. This is because when $\gamma^{'}_{0}$ is high, UAVs have to be placed close to each other to maintain the successful connection between UAVs (for backhaul connectivity), and BaG necessitates relatively large number of UAVs compared to that of the Proposed algorithm. For clarity of exposition, we employ Fig. \ref{resulting_graphs} (a) and \ref{resulting_graphs} (b) that depicts the resultant UAV placement (and multi-UAV network) corresponding to the Proposed and BaG for a simple simulation setting with $60$ ground users in an area of $75$ Km $\times$ $75$ Km. (We consider $60$ ground users for the clarity of the plot.) Here the Proposed algorithm requires $40$ UAVs whereas BaG requires $50$ UAVs.

\begin{figure}
\begin{center}
  \includegraphics[scale=0.5]{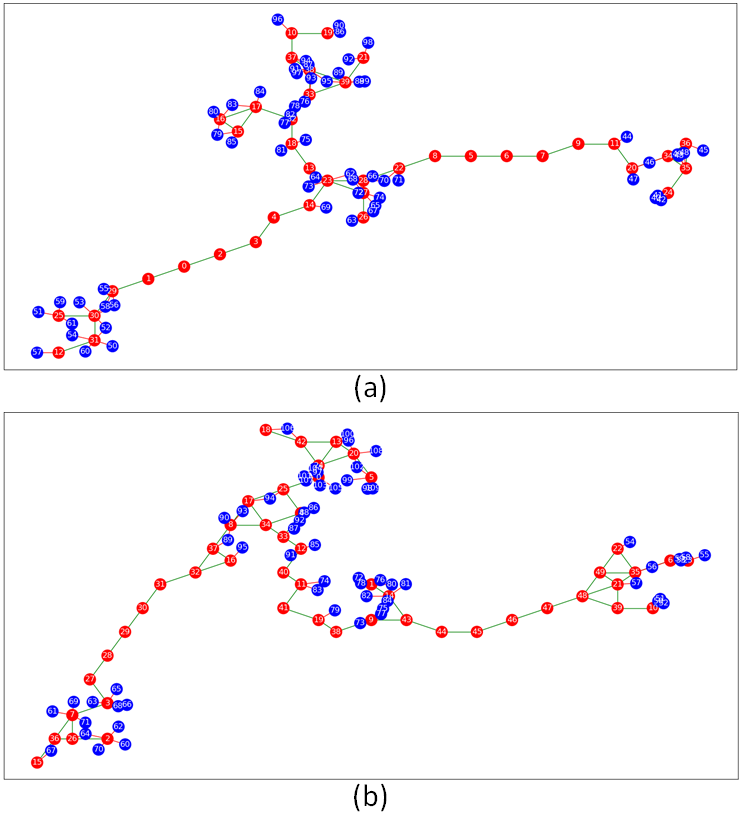}
  \caption{Resultant Multi-UAV Network topology: (a) Proposed and (b) BaG algorithms. Red nodes are UAV candidate locations, while blue ones are ground users. Green and red edges respectively represent the backhaul and access links.
  }\label{resulting_graphs}
  \end{center}
  \vspace{-0.3in}
\end{figure}

As expected, the Greedy algorithm requires a relatively smaller number of UAVs to ensure wireless coverage to all ground users, for almost all considered cases in Fig \ref{drones_vs_users}. Note that Greedy does not  guarantee backhaul connectivity, and thus does not always provide a solution to \Board problem. It is interesting to note that for lower SNR thresholds such as $10$ dB (Fig. \ref{drones_vs_users} (a)) and $15$ dB (Fig. \ref{drones_vs_users} (b)), the Proposed algorithm outperforms the Greedy algorithm by up to $13\%$ in the reduction of required UAVs. There are two reasons for this. First, since the UAV coverage radius $(R^{'})$ is large (due to lower SNR thresholds), negligible number of additional UAVs are required to ensure backhaul connectivity. Second, fewer UAVs may be required to ensure ground user coverage in case of the Proposed algorithm. Let us explain this with a simple example.  Assume 6 ground users $\{v_1,...,v_6\} \in V$ could be covered by 3 UAV candidate locations $\{k_1,k_2,k_3\} \in \mathcal{K}$ that are pairwise connected (backhaul connectivity ensured). Assume $\{v_1,v_2,v_3\}$, $\{v_2, v_3, v_4,v_5\}$, and $\{v_4,v_5,v_6\}$ are in the coverage area of $k_1$, $k_2$, and $k_3$ respectively. Therefore, by applying the Proposed algorithm, selected candidate locations will be $\{k_1,k_3\}$, while greedy algorithm selects $\{k_1,k_2,k_3\}$. 

\textbf{Varying area sizes.} Fig. \ref{drones_vs_areas} represents the variation of number of required UAVs w.r.t varying size of areas under various SNR thresholds, i.e., $\gamma^{'}_0$ = $10dB$, $15dB$, and $20dB$. As expected, the number of UAVs required by all algorithms increases with increasing area size, under all values of $\gamma^{'}_0$. This is intuitive as more UAVs will be required to provide wireless coverage to users spread out in the larger area sizes and maintain the backhaul connectivity. The number of UAVs required by the Proposed algorithm is fewer than BaG and Random algorithms, respectively by up to $17\%$ and $95\%$. The Proposed algorithm significantly outperforms both backhaul-aware algorithms, i.e., BaG and Random, for the same reasons discussed before that the Proposed algorithm solves BoaRD problem efficiently.

Interestingly, compared to the Greedy algorithm, the Proposed algorithm requires fewer number of UAVs for smaller area sizes ($< 50$ km $\times$ $50$ km) and lower $\gamma^{'}_{0}$ ($10$ and $15$ dB). However, it gradually increases afterwards mainly because of the increased number of UAVs required for ensuring backhaul connectivity. When the area becomes larger, the distances between the cluster of the users in the area are larger, and thus, more UAVs are needed to ensure the backhaul connectivity. Since, as discussed earlier, Greedy does not ensure backhaul connectivity, more UAVs are needed by Proposed and BaG algorithms compared to Greedy approach in larger areas. Also, larger $\gamma^{'}_{0}$ values results in shorter backhaul links, which increase the number of UAVs needed by Proposed and BaG algorithms in comparison to the Greedy algorithm to maintain the backhaul connectivity.

\section{Conclusions and Future Work} 
\label{secConclusion}
In this paper, we investigated the joint optimization of number, placement and backhaul connectivity of multi-UAV networks, such that, the network provides wireless coverage to all ground users in the region of operation. We formulated the above problem, named, Backhaul-and-Coverage-aware Drone Deployment (BoaRD) problem as ILP problem and showed that it is NP-Hard. Utilizing graph theoretic concepts, we proposed a low computational complexity algorithm to solve the BoaRD problem with provable performance guarantees. Our extensive simulations demonstrated the superiority of the Proposed algorithm in minimizing the number of UAVs to form a Multi-UAV Network, when compared to both backhaul-aware greedy and random algorithms for all considered scenarios. Interestingly, the Proposed algorithm even outperformed the baseline greedy algorithm (no backhaul) for higher ground user density and lower backhaul SNR thresholds, which further corroborated the efficacy of the Proposed algorithm. In future, we will explore the problem of UAV placement (with fewest number of UAVs) such that the formed Multi-UAV Network is resilient against various UAV node and link failures and provides end-to-end wireless coverage to ground users in the region of operation. Another interesting future direction is to investigate designing the Multi-UAV Networks under very high mobility ground users.

\bibliographystyle{IEEEtran}
\bibliography{myref, UAV}
	
\end{document}